\documentclass[11pt]{llncs}

\usepackage{amsmath}
\usepackage{calc,latexsym,epsfig,float}
\usepackage{amssymb}
\usepackage{amsfonts}
 
\usepackage{microtype}

\usepackage{color}
\usepackage{graphicx}
\usepackage{fullpage}
\usepackage{amsfonts}
\usepackage{latexsym,amssymb,amsmath}
\usepackage{amsmath}
\usepackage{comment}

\newtheorem{observation}{Observation}
\def\qed{\hfill\rule{2mm}{2mm}}

\newcommand{\ignore}[1]{}

\newcommand{\dense}[1]{$#1$-{dense}}
\def\cd{\dense{c}}

\newcommand{\pr}[2]{\mathop{\mathbb P}_{#1}\displaylimits \left[#2\right]}
\newcommand{\Exp}[2]{\mathop{\mathbb E}_{#1}\displaylimits \left[#2\right]}

\def\watch{\textsc{Persistence}}

\begin{document}

\title{Know When to Persist: Deriving Value from a Stream Buffer
\thanks{An extended abstract of this paper appeared in
the LNCS Springer Proceedings of the 11th  International Conference on Algorithmic Aspects of Information and Management (AAIM 2016), July 18--20 2016, Bergamo, Italy.}
}
\author{
Konstantinos Georgiou\inst{1}
\and 
George Karakostas\inst{2}
\and
Evangelos Kranakis\inst{3}\thanks{Research supported in part by NSERC Discovery grant.}
\and
Danny Krizanc\inst{4}
}

\institute{
Department of Mathematics, Ryerson University
\and
Dept. of Computing \& Software, McMaster University
\and
School of Computer Science, Carleton University
\and
Department of Mathematics \& Computer Science,
Wesleyan University
}
\maketitle

\begin{abstract}
We consider \watch, a new online problem concerning optimizing weighted observations in a stream of data when the observer has limited buffer capacity. A stream of weighted items arrive one at a time at the entrance of a buffer with two holding locations. A processor (or observer) can process (observe) an item at the buffer location it chooses, deriving this way the weight of the observed item as profit. 
The main constraint is that the processor can only move {\em synchronously} with the item stream; as a result, moving from the end of the buffer to the entrance,
it crosses paths with the item already there, and will never have the chance to process or even identify it. 
\watch\ is the online problem of scheduling the processor movements through the buffer so that its total derived value is maximized under this constraint.
We study the performance of the straight-forward heuristic {\em Threshold}, i.e., forcing the processor to "follow" an item through the whole buffer only if its value is above a threshold. We analyze both the optimal offline and Threshold algorithms in case the input stream is a random permutation, or its items are iid valued. 
We show that in both cases the competitive ratio achieved by the Threshold algorithm is at least $2/3$ when the only statistical knowledge of the items is the median of all possible values. We generalize our results by showing that Threshold, equipped with some minimal statistical advice about the input, achieves competitive ratios in the whole spectrum between $2/3$ and $1$, 
following the variation of a newly defined density-like measure of the input.
This result is a significant improvement over the case of arbitrary input streams, since in this case we show that no online algorithm can achieve a competitive ratio better than $1/2$.
 \end{abstract}

\section{Introduction}

Suppose that the Automated Quality Control (AQC) of an assembly line has the ability to check all new parts as they enter 
the assembly line. Every such check increases our quality confidence by a certain percentage, which depends on the
nature of the part itself. Now, suppose that the AQC is given the option of a second look at the same part in the
next time slot, with a similar increase in our quality confidence. The downsize of this option, is that when the AQC 
returns to the beginning of the assembly line, it will have completely missed the part immediately following the one
that was double-checked. We are looking for an algorithm to decide whether to take the option or not with every new item.
Obviously, a good strategy would strive to look twice at "low-quality" items, since that would imply the greatest increases 
to our confidence, while ``missing" only pristine-looking ones.   

This problem falls within the data stream setting: a sequence of input 
data is arriving at a very high rate,
but the processing unit has limited memory to store and process 
the input. Data stream algorithms have been
explored extensively in the computer science literature.
Typical algorithms in this area work with 
only a few passes (often just one) over the data input and use memory space 
less than linear in the input size.
Applications can be found in processing 
cell phone calls or Internet router data, 
executing Web searches, etc. (cf. ~\cite{muthukrishnan2005data,muthukrishnan2009data1}).

In this work we study a new online problem in data stream processing with limited buffer capacity.
An online stream of items (the parts in our AQC example) arrives (one item at a time) at a buffer with  
two locations $L_0, L_1$ (assembly points 1 and 2 respectively in the example above), staying at each location for one 
unit of time, in this order. A processor/observer (the AQC) can move between the two locations
{\em synchronously}, i.e., its movements happen at the same time as the items move.
This means that if the processor is processing (observing) the $i$-th item in time $t$ at $L_0$, moving to $L_1$ will
result in processing again the $i$-th item at $L_1$ in time $t+1$. On the contrary,
if the processor is processing the $i$-th item in time $t$ at $L_1$, moving to $L_0$ will
result in processing the $i+2$-th item at $L_0$ in time $t+1$; the $i+1$-th item has already
moved to $L_1$ and will leave the buffer without the processor ever encountering it! (just like the
AQC totally missed a part). We emphasize that we restrict the processor to not even know
what item it missed (i.e., cannot ``see'' into a location other than its current one).
Processing the $i$-th item (either in $L_0$ or $L_1$) produces an added value or payoff.
The processor has very limited (constant in our results) memory capacity, and cannot keep more than a few
variables or pieces of data. The problem we address is whether such a primitive processor can have a strategy to persist and observe (if possible) mostly ``good values'', especially when compared to an optimal algorithm that is aware of the input stream. We call this online problem \watch, which to the best of our knowledge is also new. 


\ignore{
Depending on the amount of a priori knowledge the processor has, there is an {\em oblivious} version of the model (i.e., the processor initially knows nothing about the input stream), and the {\em non-oblivious} version, where the processor knows that the payoffs come from a (known) set of possible values. In this work we deal with the non-oblivious setting. In addition, there are many other obvious extensions, such as different buffer sizes (instead of size 2 in this work), different payoffs for different buffer locations (instead of the uniform payoffs in this work), etc.      
}

{\bf Related Work: }
There is extensive literature on data stream algorithms.
Here the emphasis is on input data arriving at a very high rate and
limited memory to store and process 
the input (thus stressing a tradeoff between 
communication and computing infrastructure). 
A general introduction to data stream algorithms
and models can be found in
\cite{muthukrishnan2005data,muthukrishnan2009data1}.
Lower bound models for space complexity 
are elaborated in \cite{alon1996space}.
In the section on new directions for streaming models,
\cite{muthukrishnan2009data1} discusses several
alternatives for data streams for 
permutation streaming of non-repeating items \cite{ajtai2002approximate},
windowed streaming whereby the most recent past is more
important than the distant past \cite{hoffman2004location},
as well as reset model,
distributed continuous computation,
and synchronized streaming.
Applications of data stream algorithms are
explored extensively in the computer science literature
and can be found
in sampling (finding quantiles \cite{greenwald2001space}, 
frequent items \cite{manku2002approximate}, 
inverse distribution \cite{cormode2005summarizing}, and
range-sums of items \cite{alon2005estimating}).

Related to our study is the well-known {\em secretary
problem} which appeared in the late 1950s and early
1960s (see \cite{ferguson1989solved} for a historical
overview of its origins and \cite{freeman1983secretary}
which discusses several extensions). 
It is concerned with the optimal strategy or
stopping rule so as to maximize the probability of 
selecting the best job applicant assuming that 
the selection decision can be deferred to the end.
Typically we are concerned with
maximizing the probability of selecting the best job applicant;
this can be solved by a maximum selection algorithm which
tracks the running maximum,
The problem has fostered the curiosity of numerous 
researchers and studied extensively
in probability theory and decision theory.
Several variants have appeared in the scientific literature, including
on rank-based selection and cardinal payoffs \cite{bearden2006new},
the infinite secretary problem in \cite{gianini1976infinite},
secretary problem with uncertain employment in \cite{smith1975secretary},
the submodular secretary problem in \cite{bateni2013submodular},
just to mention a few. The ``secretary problem'' paradigm has important
applications in computer science of which it 
is worth mentioning the recent work of
\cite{babaioff2007matroids} which studies the relation
of matroids, secretary problems, and online mechanisms,
as well as \cite{kleinberg2005multiple}
which is investigating applications of a multiple-choice secretary 
algorithm to online auctions.
Obviously the secretary problem differs from \watch\ 
in terms of the objective function: in our case the payoff is the sum of processing payoffs,
as opposed to the maximum for the secretary problem. The two problems also differ in the
synchronicity and location
of arrivals, i.e., what can be accessed
and how it is accessed. Nevertheless, the two problems share the inherent difficulty of having to
make decisions {\em on the spot} while missing parts of the input altogether.

\subsection{High Level Summary of our Results \& Outline of the Paper}

Our primary focus is the study of the \watch\ problem, which we formally define in Section~\ref{sec: model and problem def}. Our goal is to compare the performance of any primitive (online) algorithm, which is not aware of the input stream, against the optimal offline algorithm. 
In Section~\ref{sec: persistence strategies} we present all such possible primitive algorithms that we call {\em Threshold}. Subsequently, in Section~\ref{sec: comp analysis det} we analyze the performance of any Threshold online algorithm for deterministic input streams. 
Our findings indicate that simplistic primitive algorithms are actually optimal (among all online solutions), and are off no more than 1/2 the performance of an optimal (offline) algorithm that is aware of the entire input. 
Similar to the setting of the secretary and other online decision problems, this motivates the study of \watch\ problems when the input is random, which is also our main focus. 

Our main contributions are discussed in detail in Section~\ref{sec: comp analysis det}. At a high level, we show that when the online observer (processor)  knows the median of the possible random values that can appear in the input stream, then it is possible to perform observations in a way such that the total payoff is asymptotically at least 2/3 of the optimal offline solution (Theorem~\ref{thm: comp 2/3 for all}).
Moreover, we prove that when the random input streams come from certain natural families of inputs in which the mass of possible values is concentrated in relatively few heavy items, the asymptotic performance of very primitive algorithms is nearly optimal. In fact, we parameterize the performance of online algorithms for such inputs using a proper density measure, and we show how the relative asymptotic performance changes from almost optimal (competitive ratio almost 1) to competitive ratio 2/3 (Theorem~\ref{thm: cd inputs combined}). 

The results discussed above are just the byproduct of our main technical contributions that pertain to an analytic exposition of the performance of optimal offline and any online algorithm for random inputs, parameterized by a proper statistical density-like measure on inputs. 
The two random models that we study are input streams that are either random permutations (Section~\ref{sec: permutations}) or input streams whose elements assume independent and identically distributed values (Section~\ref{sec: iid}). In each case we provide closed formulas for the performance of the optimal offline algorithm and any online algorithm (Sections~\ref{sec: performance perm} and \ref{sec: performance iid} respectively), which we think is interesting in its own right. Then we use the closed formulas to derive the promised asymptotic competitive analysis in Sections~\ref{sec: comp analysis perm} and \ref{sec: comp analysis iid} respectively. 

We emphasize that the analysis of a size-2 buffer we provide is technically involved, 
and we cannot see how it could be extended to larger buffers without considerable 
extra effort. But even for this restricted case, the problem is interesting. Indeed,
given our model of algorithms allowed (streaming algorithms with a constant-size memory that
can keep only a few variable values, i.e., memoryless), the fact that the simple threshold
algorithm achieves non-trivial improvements is already a rather surprising result.

\section{Preliminaries}

\subsection{Model \& Problem Definition}
\label{sec: model and problem def}

Assume that $n$ incoming data values $v_1, v_2, \ldots, v_{n-1}, v_n$ arrive sequentially and synchronously from the left one at a time at a processing unit consisting of two registers $L_0$ and $L_1$ which are capable of storing these values instantaneously (see also Figure~\ref{fig:ln}). 
The values pass first through location (register) $L_0$ and then through location $L_1$, before exiting. A processing unit can process (i.e., derive some payoff from or contribute some additive value to) an item either in $L_0$ or $L_1$.
The value $v_i$ derived by processing item $i$ comes from a set of possible values $a_0< a_1< \cdots < a_{k-1}$, and is independent of the location that processing happened.
The main constraint is that all processing is {\em synchronous}, i.e., at every time unit exactly one new item enters $L_0$ and the processor (observer) is allowed to either do some processing (observe) at the location it's
already in, or perform a single move (and then do processing in) to the other location. 
The other important constraint is the fact that the processor has only a constant-size memory (i.e., it has space to hold at most $O(1)$ variables) as well as it is only aware of the value of the register of its location. In particular, when processor is located at one register, it is {\em oblivious} to the value of the other register. 

\begin{figure}[!htb]
\begin{center}
\includegraphics[width=10cm]{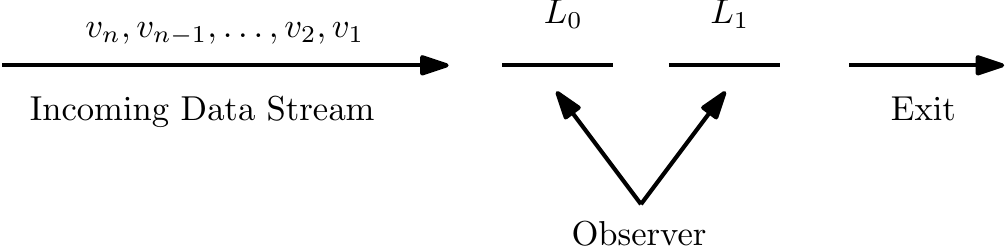}
\end{center}
\caption{Incoming values arriving sequentially and synchronously 
from the left one at a time and occupying first location $L_0$ and
then location $L_1$. An observer (processor) can look at only one of 
these two locations at a time. Data is exiting synchronously from the right.}
\label{fig:ln}
\end{figure}

More formally, our model is the following:
\begin{enumerate}
\item At time step $t=1$, the processor (observer) occupies position $L_0$, which holds value $v_1$.
\item At time step $t\geq 2$, the following take place in that order:
\begin{itemize}
\item The processor may change the location it is about to process (observe); at the same time, 
locations $L_0, L_1$ get (new) values $v_t,v_{t-1}$ respectively.
\item Processing is done at the location of the processor; the added value achieved at $t$ 
is the value of the item in that location. 
\end{itemize}
\end{enumerate}
\ignore{
The next tableau summarizes the total reward of an observer for two consecutive time steps $t,t+1$, given that $\tau_i$ denotes the current value in location $L_i$ (at time $t$), and $\tau_{-1}$ denotes the new value to enter location $L_0$ in the next step (at time $t+1$). The columns and the rows correspond to the position of the observer at time $t$ and $t+1$ respectively. 
\begin{equation}
\label{tab: sum two steps}
\begin{tabular}{c|c|c}
	& $L_0$ & $L_1$  \\
\hline
$L_0$ & $\tau_0+\tau_{-1}$		&	$\tau_1+\tau_{-1}$	\\
\hline
$L_1$ & 	$2\tau_0$	&	$\tau_1+\tau_0$	\\
\end{tabular}
\end{equation}
}
In the {\em online model} the observer is {\em not} aware of the sequence $v_1, v_2, \ldots, v_{n-1}, v_n$, rather she may only know some statistical information that requires constant memory. 
The limited memory implies a limited ability of keeping statistics or historical data, and, therefore, there is not much leeway for sophisticated processing policies. 
The (possible) movement of the observer can be determined exclusively by the current value she is observing and in particular not by the value of the location that the observer is not occupying. As a result, the only power an online algorithm has is to choose to observe a value twice in two consecutive time steps, if she thinks that this value provides high enough reward. In contrast, and in the {\em offline model}, the observer is aware of the entire sequence $v_1, v_2, \ldots, v_{n-1}, v_n$ in advance, and may choose to move between registers with no restrictions so as to maximize her total payoff. 

Our main goal is to design \watch\ strategies for the observer that maximize the total added value (or, equivalently, the {\em average} or {\em relative} added value or payoff). Our focus is to understand how the lack of information affects the performance of an oblivious online algorithm, compared to the optimal offline algorithm. The standard performance measurement that we use is the so-called {\em competitive ratio}, defined as the (worst case) ratio between the (expected - when the input stream is random) payoffs of an online and the optimal offline algorithm. It is immediate that for any input stream (even random), the competitive ratio of a fixed online algorithm is $ALG/OPT<1$, where $ALG, OPT$ are the (expected) payoff of the online and the optimal offline algorithm respectively.

\subsection{On \watch\ Strategies}
\label{sec: persistence strategies}

Given an input stream $v_1, v_2, \ldots, v_{n-1}, v_n$, the optimal solution for the offline model is straightforward; If the processor (observer) is in $L_0$, processing (observing) an item $i$ with value $v_i$, then it moves to $L_1$ only if the item that follows $i$ has a value smaller than $v_i$; If the processor is in $L_1$, processing an item $i$ with value $v_i$, then it moves to $L_0$ only if the item $i+2$ that will enter $L_0$ in the next round has a value $v_{i+2}$ greater than the value $v_{i+1}$ of the item $i+1$ currently in $L_0$. As a result, an offline and optimal observer may choose to always occupy the location (and subsequently obtain its value as a reward) that holds the maximum value that currently appears in the two locations $L_0,L_1$. Since at any step, an algorithm cannot have payoff more than the maximum value of the two registers, we conclude that
\begin{observation}
\label{obs: opt offline}
For input stream $v_1, v_2, \ldots, v_{n-1}, v_n$, and at each time step $t=2, \ldots,n$, the optimal solution of an offline algorithm incurs payoff equal to $\max\{v_t, v_{t-1}\}$. 
\end{observation}
We will invoke Observation~\ref{obs: opt offline} later, when we will derive closed formulas for the performance of the optimal offline algorithm when the values of the input stream come from certain distributions. 

Now we turn our attention to \watch\ strategies in the online model. Recall that any online algorithm is oblivious, non-adaptive and with limited memory. 
In particular, when at register $L_0$, an observer has the option to process the same value for one more time in the next step, or stay put at the register and watch in the next step the (currently unknown) value which will enter $L_0$. 
If the observer is at register $L_1$, then the possible payoff at the next step is unknown independently of the move of the observer. 
Hence it is natural to move the observer back to $L_0$, giving her the option (in the future) to observe favorable values more than once. This primitive idea gives rise to the following {\em threshold algorithms}, which are determined by a choice of threshold that dictates when a register value will be observed twice in case the observer is at register $L_0$. 
\begin{center}
\begin{tabular}{ll}
\hline
& {\bf Threshold Algorithm($T$)} \\
\hline
& {\bf Input:} a sequence of $n$ items with values $v_1, v_2, \ldots, v_{n-1}, v_n$ \\
1. &  When the processor has finished processing an item of value $\tau_0$ at $L_0$ then \\
     & 1a. {\bf if}  $\tau_0 \geq T$ {\bf then} move to $L_1$ \\
    & 1b. {\bf if} $\tau_0 < T$ {\bf then} stay at $L_0$ \\
2. & When the processor has finished processing an item at $L_1$ {\bf then} move to $L_0$ \\
\hline
\end{tabular}
\end{center} 
Our main contribution in subsequent sections is the (competitive) analysis of Threshold algorithms for various choices of thresholds. In what follows wee call the simplistic algorithm that doesn't move the processor from $L_0$ (or, equivalently, has a
threshold greater than $a_{k-1}$) {\em Naive}.

\subsection{General Input Streams}
\label{sec: comp analysis det}

In its most general version, the input stream to \watch\ is chosen by an unrestricted adversary.
Here we demonstrate that
the threshold algorithm cannot achieve a competitive ratio better than $1/2$. There are the following cases:
\begin{enumerate}
\item $a_{k-1}<T$ In this case, the processor stays always at $L_0$, and, therefore, acquires the payoff
for {\em each} item exactly {\em once}, for a total payoff of exactly $\sum_{i=1}^{n}v_i$. 
On the other hand, the optimal offline algorithm has the chance
of acquiring the payoff of the largest-value items at most {\em twice} (by processing them in both processors),
for a total payoff of, at most, $2\sum_{i=1}^{n}v_i$. Hence the competitive ratio is at least $1/2$.
\item $a_{k-2} \leq T < a_{k-1}$. In this case, the threshold algorithm always gets the payoff of an item with
value $a_{k-1}$ twice, and exactly once the values of the other items it processes in $L_0$ (obviously it
misses the items that follow immediately after the items of value $a_{k-1}$ processed in both $L_0$ and $L_1$.
It is clear that the optimal offline algorithm does the same. Therefore, the adversary will minimize this overlap between 
the threshold and optimal offline algorithms, by creating a sequence without value-$a_{k-1}$ items; this
is the same as the previous case, and the competitive ratio is at least $1/2$.
\item $T<a_{k-2}$ In this case, the adversary creates the sequence of items with values $a_{k-2}, a_{k-1}$, $a_{k-2},
a_{k-1},\ldots$. Then the relative (average) payoff for the threshold algorithm is $a_{k-2}$ (the algorithm
will always process the $a_{k-2}$ items twice, missing the more valuable $a_{k-1}$ items). The offline
optimal algorithm will behave exactly in the opposite way, achieving a relative payoff $a_{k-1}$. Hence
the competitive ratio is at most $a_{k-2}/a_{k-1}$, which can be made to be arbitrarily close to $0$.
\end{enumerate}
Therefore, the best threshold is any number greater than $a_{k-1}$, and the competitive ratio is at least $1/2$.
An upper bound of almost $1/2$ for the ratio is achieved by the input sequence with values $a_{k-2}, a_{k-1}, a_{k-2},
a_{k-1},\ldots$ and $a_{k-2}\ll a_{k-1}$.

\subsection{Competitive Analysis for Randomized Input Streams - A Summary of our Results}
\label{sec: ran model and summary of results}

The fact that, for arbitrary input streams, the Threshold algorithm cannot do better than the $1/2$ competitive ratio 
of the Naive algorithms, shows that in order for the threshold algorithm to perform better,
we need to restrict the input instances by making assumptions about the input stream. There are two 
assumptions that are common in online problems such as the secretary problem~\cite{freeman1983secretary},
or resource allocation problems~\cite{devanur2011}: One is the {\em IID} assumption, i.e., the value of each new
item is drawn independently and uniformly from the set $\{a_1,a_2,\ldots,a_{k-1}\}$. Another is the {\em random  
order} assumption, i.e., the input is a (uniformly) random permutation of $n$ items, each with its own distinct value.
In what follows, we study the threshold algorithm under these assumptions.

More formally, we study the following two random models of input streams $v_1, v_2, \ldots, v_{n-1}, v_n$:
\begin{itemize}
\item \textit{Random Permutations:} Input sequence stream is a random permutation of values $a_0\leq \ldots, \leq a_{n-1}$. \\
\item \textit{Independent and Identically Distributed Values:} Each $v_i$ assumes the value $a_j$ with probability $p_j$ independently at random, where $j=0, \leq k-1$ (note that we allow that $n=\omega(k)$).
\end{itemize} 

For both input families we assume that an online algorithm is oblivious, non-adaptive and with minimal memory, still we assume it has access in advance to some limited statistical information in order to determine a proper threshold. Our main technical contribution pertains to a detailed analysis of the performance of both the optimal offline and any Threshold online algorithm for any such random input. 
As a result we demonstrate that if the online algorithm knows the median of the set  from which the input stream elements assume values, then the competitive ratio improves significantly.
\begin{theorem}
\label{thm: comp 2/3 for all}
For any random permutation or uniform iid input stream, the online Threshold algorithm that uses as threshold the median of the values $\{a_i\}_{i}$ has (asymptotic) competitive ratio 2/3. 
\end{theorem}
We emphasize that Theorem~\ref{thm: comp 2/3 for all} is the byproduct of analytic and closed formulas that we derive for optimal offline and Threshold online algorithms, when the input stream is a random permutation (Section~\ref{sec: permutations}) or iid (Section~\ref{sec: iid}), and not necessarily uniform. 

Next we ask whether it is possible for certain families of random inputs to achieve a competitive ratio better than 2/3, and given that the online algorithm has access to some statistical information. Again, we answer this in the positive by studying generic families of instances parametrized by the relative weight of their largest values. 

\begin{definition}[\cd\ input streams]
\label{def: c-dense}
Consider a random input stream (in either the random permutation or the iid model)  whose values are chosen from $\mathcal A=\{a_1, \ldots, a_t\}$, with $a_i\leq a_{i+1}$. The input stream is called \cd\ if the total weight of the largest $\lfloor ct \rfloor$ many values of $\mathcal A$, relative to the total weight of $\mathcal A$, is equal to $1- c$, i.e. when 
\begin{equation}
\label{equa: c-dense}
1-c = \frac{\sum_{i=t- \lfloor ct \rfloor+1}^t a_i}{\sum_{i=1}^t a_i}
\end{equation}
\end{definition}
Note that, although $c$ cannot be greater than $1/2$ by definition, when $c$ is 0 or 1, then the left hand-side of~\eqref{equa: c-dense} is 1 and 0 respectively, while the right hand-side is 0 and 1 respectively. At the same time, the two sides have different monotonicity as $c$ increases, and as such the notion of \cd\ input streams is well defined. Our main contribution pertaining to the families of random inputs which are asymptotically \cd, for some $c\in (0,1/2]$, is the following. 
\begin{theorem}
\label{thm: cd inputs combined}
For any random permutation or uniform iid \cd\ input stream, the online Threshold algorithm that uses as threshold the $\lfloor c\cdot n \rfloor$ largest value of $a_i$'s has (asymptotic) competitive ratio
$
\frac12 \frac{2-c}{(1-c)(1+c)^2}.
$
\end{theorem}
Clearly, when $c$ tends to 0, the performance of our Threshold algorithms is nearly optimal for \cd\ input streams. Most notably, the worst configuration for such an input is when $c=1/2$, inducing a competitive ratio equal to 2/3 (and as already predicted by Theorem~\ref{thm: comp 2/3 for all}). The proof of Theorem~\ref{thm: cd inputs combined} for random permutations and uniform iid inputs can be found in Sections~\ref{sec: permutations} and~\ref{sec: iid} respectively. Any omitted proofs can be found in the Appendix.

\section{Random Permutation Input Streams} 
\label{sec: permutations}

In this section we study the special case of inputs that are a random permutation of $n$ items with distinct  values $a_0 < a_1 < \cdots < a_{n-1}$ (with $n \geq 2$). First we find closed formulas for the performance of the optimal offline algorithm and any Threshold online algorithm for the \watch\ problem, and then we conclude with the competitive analysis. 

\subsection{Performance of Offline and Online Algorithms for Random Permutations}
\label{sec: performance perm}

Using Observation~\ref{obs: opt offline} we show in Appendix~\ref{sec: proof of thm: opt ran per} that 
\begin{theorem}  \label{thm: opt ran per}
The relative expected payoff (asymptotic payoff per time step) of the optimal offline algorithm when the input is a random permutation is:
$
\frac 1{\binom{n}{2}} \sum_{i=1}^{n-1} i \cdot  a_i.
$
\end{theorem}
The main technical contribution of this section is the performance analysis of any Threshold online algorithm. 
\begin{theorem}   \label{thm: online permutation input}
Let $k=k(n)$ be such that $\lim_{n \rightarrow \infty} \frac kn = c \in \Theta(1)$. Let also $A^-, A^+$ denote the summation 
of the smallest $n-k$ and largest $k$ values respectively. Then the relative expected payoff of the Threshold algorithm 
(payoff per time step) when the threshold is $T:=a_{n-k}$ is:
$
\frac{A^-}{1+c}
+
\frac{2A^+}{1+c}.
$
\end{theorem}

The remaining of the section is devoted in proving Theorem~\ref{thm: online permutation input}. We will need the following random variables:
$A_i$ denotes the profit of our algorithm from value $a_i$, or in other words, the contribution of $a_i$ to the performance of the algorithm. Clearly, if $a_i \geq T$ then $A_i \in \{0,2a_i\}$, and if $a_i < T$ then $A_i \in \{0,a_i\}$. 
$V_t \in \{a_i\}_{i=0,\ldots,n-1}$ is the value that appears in position $t$ of the (random) permutation, where position 1 
is the value that will be read first ($t=1,\ldots,n$).
Finally, $O_i$ is the indicator random variable that equals 1 iff value $a_i$ is observed.

Since all values $a_i$ will appear in every permutation, we have that
\begin{equation}\label{equa: payoff online random perm}
\mbox{Expected Payoff} = 
\Exp{}{\sum_{i=0}^{n-1} A_i}
=
\sum_{i=0}^{n-1} \Exp{}{ A_i}.
\end{equation}
The contribution of each $a_i$ clearly depends on whether the value is observed. This motivates the following lemmata.

\begin{lemma}\label{lem: exp obs given 2}
For every $a_{j_0}\geq a$ and for every $a_i$ ($i \not = j_0$) we have 
\begin{equation}\label{equa: oi using complement}
\Exp{}{O_i|V_t=a_i~\&~V_{t-1}\geq T}= 1 - \Exp{}{O_j|V_t=a_i~\&~V_{t-1}= a_{j_0}}.
\end{equation}
\end{lemma}

\begin{proof}
If $a_i\geq T$, and for any fixed $a_{j_0}$, $j_0 \not= i$, we have
\begin{align*}
& \Exp{}{O_i|V_t=a_i~\&~V_{t-1}\geq T}\\
=&
\sum_{j:~a_j\geq T, j\not = i}
\pr{}{V_{t-1}=a_j |V_t=a_i~\&~V_{t-1}\geq T}
~
\Exp{}{O_i|V_t=a_i~\&~V_{t-1}= a_j} \\
=&
\frac{1}{k-1}
\sum_{j:~a_j\geq T, j\not = i}
\Exp{}{O_i|V_t=a_i~\&~V_{t-1}= a_j} 
=\Exp{}{O_i|V_t=a_i~\&~V_{t-1}= a_{j_0}}
\end{align*}
where the last equality is due to the fact that the penult expectations are all the same for all $j$ in the range of the summation. 
From the description of the threshold algorithm, and given that $V_t=a_i$ and $V_{t-1}= a_j \geq T$, we have that $O_i=1$ exactly when $O_j=0$. Therefore, we see that 
$
\Exp{}{O_i|V_t=a_i~\&~V_{t-1}\geq T}= 1 - \Exp{}{O_j|V_t=a_i~\&~V_{t-1}= a_{j_0}}
$
as we promised in~\eqref{equa: oi using complement}. The proof for $a_i<T$ is almost identical.
\qed \end{proof}

We can now compute the expected value of $O_i$ given that $a_i$ has a certain position in the permutation. 

\begin{lemma}\label{some conditional expectations}
\begin{equation}\label{equa: cond observation new}
\Exp{}{O_i|V_t=a_i}=
\left\{
\begin{array}{lll}
1-\frac{k}{n-1} f^{t-1}_{n-1,k}
& \mbox{, if }~a_i<T \\
f^t_{n,k} & \mbox{, if }~a_i\geq T \\
\end{array}
\right.
\end{equation}
where
$$
f^t_{n,k} = \frac1{\binom{n-1}{k-1}}\sum_{s=0}^{\min\{t,k\}-1} (-1)^s \binom{n-1-s}{k-1-s}.$$
\end{lemma}

\begin{proof}
From the behaviour of the threshold algorithm, it is immediate that $\Exp{}{O_i|V_t=a_i}$ depends only on whether $a_i\geq T$ or not. First we observe that 
\begin{equation}\label{equa: exp O given ai as a sum}
\Exp{}{O_i|V_t=a_i}= 
\left(
\begin{array}{rl}
& \pr{}{V_{t-1}<T|V_t=a_i}~\Exp{}{O_i|V_t=a_i~\&~V_{t-1}<T} \\
+&\pr{}{V_{t-1}\geq T|V_t=a_i}~\Exp{}{O_i|V_t=a_i~\&~V_{t-1}\geq T}
\end{array}
\right)
\end{equation}
where
$$
\pr{}{V_{t-1}\geq T|V_t=a_i}
=
\left\{
\begin{array}{lll}
\frac{k}{n-1} & \mbox{, if }~a_i<T \\
\frac{k-1}{n-1} & \mbox{, if }~a_i\geq T \\
\end{array}
\right. .
$$
The next observation is that $\Exp{}{O_i|V_t=a_i~\&~V_{t-1}<T}=1$. Indeed, if $V_{t-1}=a_j<T$, then $a_j$ is either observed or not. If it is observed, then this happens only in $L_0$, so the observer will also observe the next coming value which is $a_i$. If on the other hand $a_j$ is not observed, then necessarily the next coming value is observed. Hence, expression~\eqref{equa: exp O given ai as a sum} simplifies to 
\begin{equation}\label{equa: cond observation}
\Exp{}{O_i|V_t=a_i}=
\left\{
\begin{array}{lll}
1-\frac{k}{n-1} + \frac{k}{n-1} \Exp{}{O_i|V_t=a_i~\&~V_{t-1}\geq T} 
& \mbox{, if }~a_i<T \\
1-\frac{k-1}{n-1} + \frac{k-1}{n-1} \Exp{}{O_i|V_t=a_i~\&~V_{t-1}\geq T} & \mbox{, if }~a_i\geq T \\
\end{array}
\right. .
\end{equation}

We are now ready to justify~\eqref{equa: cond observation new} examining the two cases. 
\begin{description}
\item[Case $a_i \geq T$:] 
For every $a_i\geq T$, we set 
$
\Exp{}{O_i|V_t=a_i} := f^t_{n,k},
$
since the value is independent of $a_i$, but it is depended on the number of available values $n$, the position $t$, as well as the number of values $k$ not less than $T$. Then we observe that $\eqref{equa: oi using complement}$ of Lemma~\ref{lem: exp obs given 2} can be written as 
$\Exp{}{O_i|V_t=a_i~\&~V_{t-1}\geq T}= 1 - f^{t-1}_{n-1,k-1}.$
Continuing from~\eqref{equa: cond observation}, we see then that 
$
f^t_{n,k} = 1 - \frac{k-1}{n-1} f^{t-1}_{n-1,k-1}.
$
Given that for all $t,k\geq 1$ we have that $f^t_{n,k}=1$ whenever $t=1$ or $k=1$ the claim follows. 
\item[Case $a_i < T$:] Similar to the previous case we write $\eqref{equa: oi using complement}$ as 
$\Exp{}{O_i|V_t=a_i~\&~V_{t-1}\geq T}= 1 - f^{t-1}_{n-1,k}$
(note that in this case, and since $a_i<T$ we still have $k$ many values at least $T$ to choose from). Hence, \eqref{equa: cond observation} becomes 
$\Exp{}{O_i|V_t=a_i} = 1-\frac k{n-1} f^{t-1}_{n-1,k},$
again as promised.
\end{description} 
\qed \end{proof}

It is clear from the previous lemmata that the formulas of the payoff of online Threshold algorithms involves numerous binomial expressions, which we simplify in Appendices~\ref{sec: binomial identity lemma 1},\ref{sec: binomial identity lemma 2}. Given this quite technical work, we are ready to prove Theorem~\ref{thm: online permutation input}.
\begin{proof}[of Theorem~\ref{thm: online permutation input}]
Let $a$ denote some threshold value, such that $n-k$ many $a_i's$ are less than $a$. For every $i=0,\ldots,n-1$ we have 
\begin{equation}\label{equa: exp Ai as a sum}
\Exp{}{ A_i} = \sum_{t=1}^n \pr{}{V_t=a_i}~\Exp{}{A_i|V_t=a_i}=\frac1n \sum_{t=1}^n \Exp{}{A_i|V_t=a_i}.
\end{equation}
Using the random variables $O_i$ that indicate whether $a_i$ is observed, we have
$$
\Exp{}{A_i|V_t=a_i}
=
\left\{
\begin{array}{lll}
a_i ~\Exp{}{O_i|V_t=a_i} & \mbox{, if }~a_i<a \\
2a_i ~\Exp{}{O_i|V_t=a_i} & \mbox{, if }~a_i\geq a \\
\end{array}
\right.
$$
whose values are given by Lemma~\ref{some conditional expectations}. 
Hence, by~\eqref{equa: payoff online random perm} and \eqref{equa: exp Ai as a sum}, we have that 
\begin{align}
\mbox{Expected Payoff} 
&= 
\sum_{i=0}^{n-k-1} 
\frac1n \sum_{t=1}^n \Exp{}{A_i|V_t=a_i}
+
\sum_{i=n-k}^{n-1} 
\frac1n \sum_{t=1}^n \Exp{}{A_i|V_t=a_i} 
\notag
\\
&
=
\left(
1 - \sum_{s=0}^{k-1} (-1)^s 
\frac{\binom{n-1-s}{k-1-s}}{\binom{n}{k}}
\right) A^- +
2\left(
\sum_{s=0}^{k-1} (-1)^s 
\frac{\binom{n-s}{k-1-s}}{\binom{n}{k-1}}
\right) A^+. 
\tag{From Lemma~\ref{lem: sum of f's} in 
Appendix~\ref{sec: binomial identity lemma 1} }
\end{align}

The relative performance is obtained by dividing by $n$. According to technical Lemma~\ref{lem: asymptotics of coefs - online perm} in Appendix~\ref{sec: binomial identity lemma 2}, and given that  $\frac kn \rightarrow c$, the theorem follows.   
\qed \end{proof}

\subsection{Competitive Analysis for Random Permutations}
\label{sec: comp analysis perm}

We can now prove Theorems~\ref{thm: comp 2/3 for all} and \ref{thm: cd inputs combined} pertaining to random permutations. Suppose that the Threshold algorithm chooses threshold value $T$ equal to the $\bar k$ largest element of the value $a_i$. Denote by $A$ the sum of all values $a_i$, and by $L_{\bar k}$ the sum of the $\bar k$ largest values of them. Abbreviate also $\bar k/n$ by $c$. 

Theorem \ref{thm: online permutation input} applies with $T:=a_{n-\bar k}$, to give
(asymptotically) that \begin{equation}
\label{equa: perf online ran per equiv}
ALG=\frac1{1+c}\frac1nA+\frac c{1+c}\frac1{\bar k} L_{\bar k}.
\end{equation}
At the same time, Theorem \ref{thm: opt ran per} implies that for the optimal offline algorithm we have
\begin{equation}  \label{bopt}
OPT = \frac 1{\binom{n}{2}} \sum_{i=0}^{n-1} i \cdot  a_i
\leq 
\frac 1{\binom{n}{2}} 
\left(
(n-k) \sum_{i=0}^{n-\bar k-1}  a_i 
+ 
n \sum_{i=n-\bar k}^{n-1}  a_i 
\right)
=
\frac2n\left( (1-c)A+cL_{\bar k}\right). 
\end{equation}

\begin{proof}[of Theorem~\ref{thm: comp 2/3 for all} for Random Permutations]
When the Threshold value is the median, we have that $c=1/2$. Using the bounds~\eqref{equa: perf online ran per equiv} and~\eqref{bopt} it is straightfoward that $ALG, OPT$ are indeed within 2/3 of each other as promised. 
\qed \end{proof}

\begin{proof}[of Theorem~\ref{thm: cd inputs combined} for Random Permutations]
When the input is a \cd\ stream, the Threshold algorithm can choose $\bar k$ satisfying $1-\bar k/n = 1-c = L_{\bar k}/A$. But then the competitive ratio becomes
$$
\frac
{
ALG
}
{
OPT
}
\stackrel{\eqref{equa: perf online ran per equiv},\eqref{bopt}}
\geq 
\frac12 \cdot \frac1{1+c} \cdot 
\frac
{
1+\frac{L_{\bar k}}A
}
{(1-c)+c \frac{L_{\bar k}} A}
=
\frac12 \cdot \frac{2-c}{(1-c)(1+c)^2}.
$$ 
\qed \end{proof}

\section{Random iid-Valued Input Streams}
\label{sec: iid}

In this section we study the special case of inputs streams whose elements are iid valued. As per the description of the model in Section~\ref{sec: ran model and summary of results}, we assume that the value $v_i$ of the $i$-th input item of the stream is an independent random variable assuming a value $a_0 < a_1 < \cdots < a_{k-1}$ (with $k \geq 2$) with probability $p_0, p_1,\ldots, p_{k-1}$ respectively (i.e., $\Pr [v_i= a_j] = p_j$). 

\subsection{Performance of Offline and Online Algorithms for iid-Valued Streams}
\label{sec: performance iid}

Using Observation~\ref{obs: opt offline}, we compute in Appendix~\ref{sec: proof of off iid app} the asymptotic payoff of the optimal offline algorithm.
\begin{theorem}  \label{thm:offiid}
The relative expected payoff (asymptotic payoff per time step) of the optimal offline algorithm when the input is a random i.i.d. sequence is
$
\sum_{i=0}^{k-1} p_i a_i
+
\sum_{i=0}^{k-1}\sum_{j=i+1}^{k-1}p_ip_j(a_j-a_i).
$
\end{theorem}

The remaining of this section is devoted in determining the asymptotics of any Threshold algorithm. 

\begin{theorem}  \label{thm: online iid input}
The relative expected payoff of the Threshold algorithm (asymptotic payoff per time step) that uses threshold $T=a_{r}$ and when the input is a random i.i.d. is
$
\frac
{\sum_{i=0}^{r-1} p_i a_i + 2\sum_{i=r}^{k-1} p_i a_i}
{\sum_{i=0}^{r-1} p_i + 2\sum_{i=r}^{k-1} p_i}.
$
\end{theorem}

\begin{proof}
In what follows we introduce abbreviations $Avg:= \sum_{i=0}^{k-1} p_ia_i$ and  $P:= \sum_{i=r}^{k-1} p_j$. Let also $Y_i$ be the random variable such that $Y_i=b$ indicates that, at time $i$, the observer is at $L_b$, $b\in \{0,1\}$. Let also $q_i:=\pr{}{Y_i=0}$. By definition, $q_0 = 1$. Next we observe that 
$$
1 - q_{i+1}
= \pr{}{Y_{i+1}=1}
= \pr{}{Y_{i}=0 ~\&~X_i \geq T }
= \pr{}{X_i \geq T~|~Y_{i}=0 }~ \pr{}{Y_{i}=0}
= P q_i .
$$
Technical Lemma~\ref{ref: rec solve for iid} in Appendix~\ref{sec: sol to recurrence for iid} implies that 
\begin{equation}
\label{equa: prob of position}
q_i =
\frac{1-(-1)^iP^i}{1+P}.
\end{equation}
Next we observe that $\Exp{}{X_i~|~Y_i=0}=Avg$. Also, if we set $Avg^+:=\sum_{s=r}^{k-1}a_sp_s$ we see that
$$
\Exp{}{X_i~|~Y_i=1}
=
\Exp{}{X_i~|~Y_{i-1}=0 ~\&~ X_{i-1}\geq T} = \frac{Avg^+}{P}.
$$
We now compute
\begin{align*}
\Exp{}{\sum_{i=1}^{n}X_i} &= 
\sum_{i=1}^{n} \Exp{}{X_i} 
=
\sum_{i=1}^{n} 
\left(
\pr{}{Y_i=0} \Exp{}{X_i~|~Y_i=0} 
+
\pr{}{Y_i=1} \Exp{}{X_i~|~Y_i=1} 
\right)\\
&\stackrel{\eqref{equa: prob of position}}{=}
\left(
\frac{n}{1+P}+\frac{P+(-P)^{n+1}}{(1+P)^2}
\right) \cdot Avg
+
\left(
\frac{n}{1+P}+\frac{(-P)^n-1}{(1+P)^2}
\right) \cdot Avg^+.
\end{align*}
Dividing the last quantity by $n$, and taking the limit $n\rightarrow \infty$ gives the promised formula. 
\qed \end{proof}

\subsection{Competitive Analysis for Uniform iid-Valued Input Streams}
\label{sec: comp analysis iid}

Note that the formulas derived in Section~\ref{sec: performance iid} hold for all iid-valued input streams. In this section we provide competitive analysis for input streams that are uniformly valued, i.e. when $p_i=\frac1k$, for all $i=0,\ldots,k-1$. That would be Theorems~\ref{thm: comp 2/3 for all} and \ref{thm: cd inputs combined} pertaining to uniform iid-valued random input streams. 

As before, denote by $A$ the sum of all values $a_i$, and by $L_{\bar r}$ the sum of the $\bar r$ largest values of them. Abbreviate also $\bar r/n$ by $c$. Suppose also that the Threshold algorithm uses as threshold the $\bar r$-th largest value of the $a_i$'s. We use Theorem \ref{thm:offiid} to find an upper bound for the performance of the offline algorithm:
\begin{equation}  
OPT = \frac1k A + \frac1{k^2}\sum_{i=0}^{k-1}\sum_{j=i+1}^{k-1}(a_j-a_i) =  \frac1k A + \frac1{k^2} \sum_{i=0}^{k-1} (2i-k+1) a_i
\leq 2\left( \frac{k-\bar r+1/2}{k^2}A + \frac{\bar r}{k^2}L_{\bar r} \right). 
\label{bopt iid}
\end{equation}
Next, using Theorem \ref{thm: online iid input} (which is written for threshold value $a_r=a_{k-1-\bar r}$) we obtain that for the Threshold algorithm 
\begin{equation}   \label{iidalg}
ALG = 
\frac 1k \cdot \frac{A + L_{\bar r}}{1+\bar r/k}.
\end{equation}

\begin{proof}[of Theorem~\ref{thm: comp 2/3 for all} for Uniform iid-Valued Streams]
When the Threshold value is the median, we have that $\bar r/k=1/2$. Using bounds~\eqref{bopt iid} and~\eqref{iidalg}, it is straightfoward then to see that $ALG, OPT$ are indeed within 2/3 of each other as promised. 
\qed \end{proof}

\begin{proof}[of Theorem~\ref{thm: cd inputs combined} for Uniform iid-Valued Streams Random Permutations]
When the input is a \cd\ stream, the Threshold algorithm can choose $\bar r$ satisfying $1-\bar r/n = 1-c = L_{\bar r}/A$. But then the competitive ratio becomes
$$
\frac
{
ALG
}
{
OPT
}
\stackrel{\eqref{bopt iid},\eqref{iidalg}}
\geq 
\frac12 \cdot \frac1{1+c} \cdot 
\frac
{
1+\frac{L_{\bar r}}A
}
{(1-c+o(c) )+c \frac{L_{\bar r}} A}
\rightarrow 
\frac12 \cdot \frac{2-c}{(1-c)(1+c)^2}.
$$
\qed\end{proof}


\section{Open Problems} \label{conclu:sec}

As described in the introduction, our model can be extended in many different ways. An obvious extension is
to have a bigger buffer, i.e., $k>2$ locations $L_0, L_1,\ldots,L_{k-1}$. In this case, there are different
possibilities of moving the processor within the buffer: a {\em single jump} model would require the processor
to always jump to $L_0$, while a {\em local jump} model would allow the processor to move close to its 
current location. Another obvious extension would be to consider general payoffs, i.e., allowing an item to
have different values in different buffer locations. Also, we leave open the potential increase in the power
of the processor if it is allowed to know the item it's going to miss in $L_0$ (if it moves to $L_0$ from $L_1$
in the next time slot).    

The threshold algorithm is probably the simplest algorithm one can use to tackle \watch. The obvious question
is whether there are better algorithms for the non-oblivious setting. Also, are there {\em upper bounds} that can be shown?    
In the oblivious setting, it is obvious that the thresholds we calculated above do not apply since we do
not know the payoffs ahead of time. In that setting, it is natural to consider adaptive algorithms, probably using
a prefix of the input in order to `learn' something about it before employing a threshold-like or some other strategy.



\bibliographystyle{plain}
\bibliography{refs}



\appendix

\section{Parts Omitted from Section~\ref{sec: performance perm}}

\subsection{Proof of Theorem \ref{thm: opt ran per}}
\label{sec: proof of thm: opt ran per}

If $X_i$ is the random variable whose value is the profit of the optimal algorithm at time $i$, we need to calculate $\Exp{}{\sum_{i=1}^{n+1} X_i}= \sum_{i=1}^{n+1} \Exp{}{X_i}$. When calculating the relative expected payoff, the extreme case observations $X_1, X_{n+1}$ have (asymptotically) 0 contribution. So we may focus on a fixed and arbitrary time step $i$ and evaluate $\Exp{}{X_i}$. 

At time $i$, let $T_0,T_1$ denote the random variables that are equal to values in the  two windows. For the optimal algorithm, we have that $X_i = \max\{T_0, T_1\}$. Since the input is a random permutation, and for all $i\not = j$, we observe that
$$
\pr{}{T_0=a_i ~\& ~T_1=a_j} = \frac{(n-2)!}{n!}=\frac{1}{n(n-1)}.
$$
Hence, using also Observation~\ref{obs: opt offline}, we have
\begin{align*}
\Exp{}{X_i}
&=
\sum_{i=0}^{n-1} ~~ \sum_{j \in \{0,1,\ldots,n-1\} \setminus \{i\}}
\pr{}{T_0=a_i ~\& ~T_1=a_j}
\max\{a_i, a_j\} \\
&=
\frac{1}{n(n-1)}
\left(
\sum_{i=0}^{n-1} \sum_{j =0}^{i-1}  \max\{a_i, a_j\}
+
\sum_{i=0}^{n-1} \sum_{j =i+1}^{n-1} \max\{a_i, a_j\}
\right) \\
&=
\frac{1}{n(n-1)}
\left(
\sum_{i=0}^{n-1} \sum_{j =0}^{i-1}  a_i
+
\sum_{i=0}^{n-1} \sum_{j =i+1}^{n-1} a_j
\right) \\
&=
\frac{1}{n(n-1)} \sum_{i=1}^{n-1} 2i \cdot  a_i.
\end{align*}

\subsection{Combinatorial Identities - Part I}
\label{sec: binomial identity lemma 1}

\begin{lemma}\label{lem: sum of f's}
\begin{align*}
(a)~~ &
\sum_{t=1}^n  f^{t}_{n,k} =
n 
\sum_{s=0}^{k-1} (-1)^s 
\frac{\binom{n-s}{k-1-s}}{\binom{n}{k-1}} \\
(b)~~ &
\sum_{t=1}^n \left( 1-\frac{k}{n-1} f^{t-1}_{n-1,k} \right)
=
n - n \sum_{s=0}^{k-1} (-1)^s 
\frac{\binom{n-1-s}{k-1-s}}{\binom{n}{k}}.
\end{align*}
\end{lemma}

\begin{proof}

\begin{align*}
(a)~~\binom{n}{k-1} \sum_{t=1}^n  f^{t}_{n,k} 
&=
\sum_{t=1}^n
\sum_{s=0}^{\min\{t,k\}-1} (-1)^s \binom{n-1-s}{k-1-s}
\\
&=
\sum_{t=1}^{k}
\sum_{s=0}^{t-1} (-1)^s \binom{n-1-s}{k-1-s}
+
\sum_{t=k+1}^{n}
\sum_{s=0}^{k-1} (-1)^s \binom{n-1-s}{k-1-s}
\\
&=
\sum_{s=0}^{k-1} (-1)^s (k-s)\binom{n-1-s}{k-1-s}
+
(n-k)\sum_{s=0}^{k-1} (-1)^s \binom{n-1-s}{k-1-s} 
\\
&=
\sum_{s=0}^{k-1} (-1)^s (n-s)\binom{n-1-s}{k-1-s}\\
&=
n \sum_{s=0}^{k-1} (-1)^s 
\binom{n-s}{k-1-s}
\end{align*}
\begin{align*}
(b)~~\sum_{t=1}^n \left( 1-\frac{k}{n-1} f^{t-1}_{n-1,k} \right)
&=
n - \frac{k}{n-1} \sum_{t=1}^n f^{t-1}_{n-1,k}  \\
&\stackrel{(f^0_{n,k}=0)}=
n - \frac{k}{n-1} \sum_{t=2}^{n} f^{t-1}_{n-1,k}  \\
&=
n - \frac{k}{n-1} \sum_{t=1}^{n-1} f^{t-1}_{n-1,k}  \\
&\stackrel{(a)}=
n - k \sum_{s=0}^{k-1} (-1)^s 
\frac{\binom{n-1-s}{k-1-s}}{\binom{n-1}{k-1}} \\
&=
n - n \sum_{s=0}^{k-1} (-1)^s 
\frac{\binom{n-1-s}{k-1-s}}{\binom{n}{k}}. \\
\end{align*}
\qed \end{proof}

\subsection{Combinatorial Identities - Part 2}
\label{sec: binomial identity lemma 2}

\begin{lemma}\label{lem: asymptotics of coefs - online perm}
For every $k=k(n)$, let $\lim_{n\rightarrow \infty} \frac kn = c \in \Theta(1)$. Then 
$$
\lim_{n\rightarrow \infty} \frac{1}{\binom{n}{k-1}} \sum_{s=0}^{k-1} (-1)^s \binom{n-s}{k-1-s}=\frac 1{1+c}.
$$

\end{lemma}

\begin{proof}
We show that for every $\epsilon>0$ (that can be chosen to be arbitrarily small), such that $\log{\epsilon}/\log{c}$ is an even integer, we have 
\begin{equation}   \label{limit}
\frac {1}{1+c} 
-\frac {\epsilon}{1+c} 
\leq \lim_{n\rightarrow \infty} \frac{1}{\binom{n}{k-1}} \sum_{s=0}^{k-1} (-1)^s \binom{n-s}{k-1-s}\leq 
\frac{1}{1+c}+
\frac{\sqrt{\epsilon}}{1-c}.
\end{equation}

Indeed, consider the odd constant $r:=\log{\epsilon}/\log{c}-1$. Then we have 
\begin{align*}
& \lim_{n\rightarrow \infty} \frac{1}{\binom{n}{k-1}} \sum_{s=0}^{k-1} (-1)^s \binom{n-s}{k-1-s} \\
&=
\lim_{n\rightarrow \infty}
\underbrace{
 \frac{1}{\binom{n}{k-1}} \sum_{s=0}^{r} (-1)^s \binom{n-s}{k-1-s} }_{A(n)}
+
\lim_{n\rightarrow \infty}
\underbrace{
\frac{1}{\binom{n}{k-1}} \sum_{s=r+1}^{k-1} (-1)^s \binom{n-s}{k-1-s} 
}_{B(n)}.
\end{align*}

Since $A(n)$ is a finite sum we have 
\begin{align}\notag
\lim_{n\rightarrow \infty} A(n) &= 
\sum_{s=0}^{r} (-1)^s \left( \lim_{n\rightarrow \infty} \frac{1}{\binom{n}{k-1}} \binom{n-s}{k-1-s} \right)\\
&= \notag
\sum_{s=0}^{r} (-1)^s c^s \\
&= \notag
\frac{1+(-1)^r c^{r+1}}{1+c}\\
&=\label{equa: first constant summands}
\frac{1-\epsilon}{1+c}.
\end{align}

Now, without loss of generality assume that $k$ is even (otherwise the last summand is positive and our bound below is still valid). We observe that 
\begin{align*}
B(n)&=
\sum_{s=r+1}^{k-1} (-1)^s \frac{\binom{n-s}{k-1-s}}{{\binom{n}{k-1}}}\\
&=
\sum_{s=(r+1)/2}^{(k-2)/2} 
\left(
\frac{\binom{n-s}{k-1-s}}{{\binom{n}{k-1}}} 
-
\frac{\binom{n-s-1}{k-2-s}}{{\binom{n}{k-1}}} 
\right)\\
&=
\left(1-\frac{k+1}n\right)
\sum_{s=(r+1)/2}^{(k-2)/2} 
\prod_{j=0}^{s+1} \frac{k-1-j}{n-1-j}.
\end{align*}
It is easy to see that $\frac{k-1-j}{n-1-j}$ are decreasing with $j$, and since each term is positive, we have 
\begin{equation}\label{equa: Bn upper bound}
0< B(n) < 
\left(1-\frac{k+1}n\right)
\sum_{s=(r+1)/2}^{(k-2)/2} 
\left( \frac{k-1}{n-1} \right)^{s+2}.
\end{equation}
The non-negativity of $\lim_{n\rightarrow \infty} B(n)$, together with~\eqref{equa: first constant summands}, imply the lower bound of \eqref{limit}. As for the upper bound, we introduce the shorthand $q:=(k-1)/(n-1)$, and we see that~\eqref{equa: Bn upper bound}, after we compute the sum in the right-hand side, implies that 
$$
B(n) < \frac{q^2}{1-q}\left( q^{(r+1)/2} - q^{k/2} \right) < \frac{q^{(r+5)/2}}{1-q}.
$$
Therefore 
$$
\lim_{n\rightarrow \infty} B(n) < 
\frac{\left(\lim_{n\rightarrow \infty}\frac{k-1}{n-1}\right)^{(r+5)/2}}{1-\lim_{n\rightarrow \infty}\frac{k-1}{n-1}}
=
\frac{c^{(r+5)/2}}{1-c} = \frac{c^2 \sqrt{\epsilon}}{1-c}.
$$
Combining the above, and given that $c\leq 1$, we conclude that \eqref{limit} holds.
\qed \end{proof}

\ignore{
\begin{lemma}  \label{lem:simple}
\begin{eqnarray}
\frac 1n \sum_{t=1}^n (1-\frac{k}{n-1} f_{n-1,k}^{t-1}) &=& 1-\frac 1{\binom{n}{k}} A_k  \label{an} \\
\frac 2n \sum_{t=1}^n f_{n,k}^t &=& \frac 2{\binom{n}{k}} \frac{n-k+1}{k} \left[ \frac{3n+1-(-2)^{k-2}}{3k} A_k+
\frac{(-2)^{k-2}}{k}\right]   \label{bn} 
\end{eqnarray}
where
\[ A_k=\sum_{s=0}^{k-1} (-1)^s \binom{n-1-s}{k-1-s}. \]
\end{lemma}
\begin{proof}
We set 
\begin{equation}  \label{B}
B=\sum_{s=0}^{k-1} (-1)^s \binom{n-s}{k-1-s}.
\end{equation}
Then we have the following
\begin{equation}  \label{BB}
B = \frac{n-k+1}{k} \left[A_k+\sum_{s=0}^{k-1} (-1)^s \binom{n-1-s}{k-2-s}\right] = \frac nk A_k - \frac 1k \sum_{s=0}^{k-1} s(-1)^s \binom{n-1-s}{k-1-s}.
\end{equation}
We define
\begin{equation}   \label{X}
X_1 = -1, X_k = \sum_{s=0}^{k} s(-1)^s \binom{n-1-s}{k-s}.
\end{equation}
After some algebra, we get the following recurrence:
\[ X_k+X_{k-1}+2A_k=0, X_1=-1 \]
which gives the solution
\[ X_k = \frac{(-2)^{k-1}-1}{3} A_k -(-2)^{k-1},\ k\geq 1\]
and by substituting in \eqref{BB}, and then in Lemma \ref{lem: sum of f's} we get the lemma.

Note that $A_k$ satisfies the recurrence
\begin{equation}   \label{eqn:A}
A_{k-1} = \frac{(3k-2-(-2)^{k-2})A_k+3(-2)^{k-2}}{3(n+1)}
\end{equation}
\qed \end{proof}
}

\section{Parts Omitted from Section~\ref{sec: performance iid}}

\subsection{Proof of Theorem \ref{thm:offiid}}
\label{sec: proof of off iid app}

A random i.i.d. input realizes into the sequence $a_{i_{1}}a_{i_{2}}\ldots a_{i_{n-1}}$ with probability $\prod_{j=1}^n p_{i_j}$, where $i_j \in \{0,k-1\}$, $j=1,\ldots, n$. 
By Observation~\ref{obs: opt offline}, it follows that the expected profit between time 2 and $n$ equals
\begin{equation}\label{equa: axp payoff offline 2 to n}
\sum_{i_n=0}^{k-1} 
\sum_{i_{n-1}=0}^{k-1} 
\ldots
\sum_{i_1=0}^{k-1} 
\prod_{j=1}^n p_{i_j}
\sum_{t=2}^n \max\{a_{i_t}, a_{i_{t-1}} \}.
\end{equation}
We focus at the case $t=2$, since it is immediate from the above formula that calculations will be identical for any $t\in \{2,\ldots, n\}$. The summand of \eqref{equa: axp payoff offline 2 to n} corresponding to $t=2$ equals
\begin{align*}
& \sum_{i_n=0}^{k-1} 
\sum_{i_{n-1}=0}^{k-1} 
\ldots
\sum_{i_1=0}^{k-1} 
\prod_{j=1}^n p_{i_j}
\max\{a_{i_2}, a_{i_{1}} \}
\\
=&
\sum_{i_n=0}^{k-1} 
\sum_{i_{n-1}=0}^{k-1} 
\ldots
\sum_{i_3=0}^{k-1} 
\prod_{j=3}^n p_{i_j}
\left(
\sum_{i_2=0}^{k-1} 
\sum_{i_1=0}^{k-1} 
p_{i_1} p_{i_2}
\max\{a_{i_2}, a_{i_{1}} \}
\right)
\\
=&
\sum_{i=0}^{k-1} 
\sum_{j=0}^{k-1} 
p_{i} p_{j}
\max\{a_{i}, a_{j} \}
\\
=&
\sum_{i=0}^{k-1} 
\sum_{j=0}^{i} 
p_{i} p_{j}
\max\{a_{i}, a_{j} \}
+
\sum_{i=0}^{k-1} 
\sum_{j=i+1}^{k-1} 
p_{i} p_{j}
\max\{a_{i}, a_{j} \}
\\
=&
\sum_{i=0}^{k-1} 
\sum_{j=0}^{i} 
p_{i} p_{j}
a_i
+
\sum_{i=0}^{k-1} 
\sum_{j=i+1}^{k-1} 
p_{i} p_{j}
a_j
\\
=&
\sum_{i=0}^{k-1} 
\sum_{j=0}^{k-1} 
p_{i} p_{j}
a_i
-
\sum_{i=0}^{k-1} 
\sum_{j=i+1}^{k-1} 
p_{i} p_{j}
a_i
+
\sum_{i=0}^{k-1} 
\sum_{j=i+1}^{k-1} 
p_{i} p_{j}
a_j
\\
=&
\sum_{i=0}^{k-1} 
p_{i}
a_i
+
\sum_{i=0}^{k-1} 
\sum_{j=i+1}^{k-1} 
p_{i} p_{j} (a_j-a_i).
\end{align*}
The last formula concludes the theorem.

\subsection{Solution to a Recurrence}
\label{sec: sol to recurrence for iid}

\begin{lemma}
\label{ref: rec solve for iid}
The solution to the recurrence $1 - q_{i+1} = P q_i$ is given by the formula $q_i =
\frac{1-(-1)^iP^i}{1+P}$.
\end{lemma}

\begin{proof}
This recurrence can
be solved using generating functions. Indeed, let us define the
function $f(x) := \sum_{i\geq 0} q_i x^i$. If we multiply the recurrence
by $x^{i+1}$ we see that $q_{i+1}x^{i+1}  + A q_i x^{i+1} = x^{i+1}$.
Summing all these recurrences for $i \geq 0$ we conclude that
$$
\sum_{i \geq 0} q_{i+1}x^{i+1}  + A x \sum_{i \geq 0} q_i  x^{i} 
= \sum_{i \geq 0} x^{i+1} .
$$
This last equation is easily seen to be equivalent to
$f(x) + Ax f(x) = \frac{x}{1-x}$. It follows that
$f (x) = \frac{x}{(1-x) (1+ Ax)}
$
from which we easily derive that
$
q_i =
\frac{1-(-1)^iP^i}{1+P}
$, as wanted. 
\qed \end{proof}

\end{document}